\documentclass[runningheads,a4paper,11pt]{llncs}

\usepackage[T1]{fontenc}
\usepackage[utf8]{inputenc}
\usepackage[dvips]{graphicx}
\usepackage{moreverb}
\usepackage{psfrag}
\usepackage{amsmath}
\usepackage{amsfonts}
\usepackage{amssymb}
\usepackage{hyperref}
\usepackage{authblk}

\usepackage{color}
\usepackage{framed}

\ifx\suppress\undefined
\newcommand{\TODO}[1]{
\typeout{WARNING!!! there is still a TODO left}
\marginpar{\textbf{!TODO: }\emph{#1}}
}
\else
\newcommand{\TODO}[1]{}
\fi

\ifx\suppress\undefined
\newenvironment{todo}[1]{\noindent\rule{.3\textwidth}{1pt}\TODO{#1}\\}{\\\rule{.3\textwidth}{1pt}}
\else

\fi

\ifx\suppress\undefined
\newcommand{\NOTE}[1]{
\typeout{WARNING!!! there are still DRAFT NOTES left}
\marginpar{!DRAFT}\emph{\textbf{DRAFT NOTES:} #1}
}
\else
\newcommand{\NOTE}[1]{}
\fi

\newcommand{\comp}[1]{{#1}^c}

\newcommand{\ket}[1]{\left| #1 \right\rangle}

\title{Reversibility in the Extended Measurement-based~Quantum~Computation}
\author{Nidhal Hamrit\inst{1,3}\and Simon Perdrix \inst{2,3}}
\institute{Telecom ParisTech \and CNRS \and LORIA, UMR 7503, CNRS, Inria, Universit\'e de Lorraine, Nancy}
\date{}

\begin{document}

\maketitle

\begin{abstract}
When applied on some particular quantum entangled states, measurements are universal for quantum computing. In particular, despite the fondamental probabilistic evolution of quantum measurements, any unitary evolution can be simulated by a measurement-based quantum computer (MBQC). We consider the extended version of the MBQC where each measurement can occur not only in the (X,Y)-plane of the Bloch sphere but also in the (X,Z)- and (Y,Z)-planes. The existence of a gflow in the underlying graph of the computation is a necessary and sufficient condition for a certain kind of determinism. We extend the focused gflow (a gflow in a particular normal form) defined for the (X,Y)-plane to the extended case, and we provide necessary and sufficient conditions for the existence of such normal forms. 
%
\end{abstract}

\section{Introduction}

Performing one-qubit measurements on an initially entangled state called  \emph{graph state} \cite{HEB04} is a universal model for quantum computation introduced by Raussendorf and Briegel \cite{RB01,RBB03}. This model is very promising  for the physical implementation of a quantum computer \cite{Wetal,Petal}. The measurement-calculus \cite{DKP,DKPP} is a formal framework for measurement-based quantum computation. In the original model introduced by Briegel and Raussendorf, all measurements are applied in the so called $\{X,Y\}$-plane of the Bloch sphere, however the model can be extended to other planes, namely $\{X,Z\}$ and $\{Y,Z\}$-planes. For instance, measurements in the $\{X,Z\}$-planes are universal \cite{MP} for quantum computation, with the particular property that only real numbers are used in this case. The Extended  Measurement-Calculus \cite{DKPP} is an extension of the Measurement-Calculus in which the three possible planes of measurement are available.

The question of the reversibility is central in measurement-based quantum computation since the key ingredient of this model -- the quantum measurement -- has a fundamentally probabilistic evolution. Reversibility is essential for the simulation of quantum circuits, and as a consequence for the universality of the model. For deciding whether an initial resource (a graph state) can be used to implement a  reversibile evolutions, a graphical condition called \emph{gflow} has been introduced \cite{DK,BKMP}. 

Gflow is not unique in general. In the non-extended case a focused gflow \cite{MMPST}  is nothing but a gflow in some normal form. We consider three natural extensions of the focused gflow for the extended measurement based quantum computation and we study in which cases these normal forms exist. 


\section{Extended Measurement-based Quantum Computation}

In this section, a brief description of the extended measurement-based quantum computation is given, a more detailed introduction can be found in \cite{DKP,DKPP}. An measurement-based quantum computation (MBQC) is: 
\begin{itemize}
\item[(i)] {\bf Initialisation.} An \emph{open graph} $(G,I,O)$ which describes the initial entanglement ($G=(V,E)$ is a simple undirected graph), the inputs ($I\subseteq V$) and outputs ($O\subseteq V$) of the computation. The initial entanglement is obtained by applying the following preparation map $N$ which associates with every arbitrary input state located on the input qubits the initial entangled state of the MBQC:
\begin{eqnarray*}N : \mathbb C^{\{0,1\}^I}&\to& \mathbb C^{\{0,1\}^{V}}\\\ket x &\mapsto&   \frac{1}{\sqrt{2^{|\comp{I}|}}} \sum_{y\in \{0,1\}^{\comp{I}}} (-1)^{|G[x,y]|}\ket{x,y}\end{eqnarray*} 
where $G[x,y]$ denotes the subgraph of $G$ induced by the supports of $x$ and $y$ and $|G[x,y]|$ its size. 
In other words $|G[x,y]|$ is the number of edges $(u,v)\in E$ such that $(x(u){=}1 \vee y(u) {=} 1) \wedge (u(v){=}1 \vee y(v){=}1)$; 
~\\

\item[(ii)] {\bf Measurements.} For every non output qubit $u\in \comp{O}$, $\alpha(u)\in [0,2\pi)$ and two distinct Pauli operators $\lambda_1(u)$, $\lambda_2(u) \in \{X,Y,Z\}$ 
 describe  the \emph{plane} $\{\lambda_1(u),\lambda_2(u)\}$ and the \emph{angle} $\alpha(u)$ according to which the qubit $u$ is measured i.e., $u$ is measured according to the observable $$\cos(\alpha(u))\lambda_1(u)+\sin(\alpha(u))\lambda_2(u)$$ 
 Measurement of qubit $u$ produces a classical outcome $(-1)^{s_u}$ where $s_u\in \{0,1\}$ is called \emph{signal}, or simply \emph{classical outcome} with a slight abuse of notation; 
 
 ~\\
 
 \item[(iii)] {\bf Corrections.} Two maps $\mathtt x, \mathtt z: \comp{O} \to 2^{V}$ called \emph{corrective maps}. Corrections work as follows: for every non output qubit $u$, the measurement of qubit $u$ is followed by the application of $X^{s_u}$ on the qubits in $\mathtt x(u)$ and $Z^{s_u}$ on the qubits in $\mathtt z (u)$. A vertex $v\in \mathtt x(u)\cup \mathtt z(u)$ is called a \emph{corrector} of $u$.
 The maps $\mathtt x, \mathtt z$ should be \emph{extensive} in the sense that there exists a partial order $\prec$ over the vertices of the graph s.t. any corrector $v$ of a vertex $u$ is larger than $u$, i.e. $v\in \mathtt x(u)\cup \mathtt  z(u)$ implies  $u\prec v$.  The extensivity of $\mathtt x$ and $\mathtt z$ guarantees that the corrections are applied on qubits which are no yet measured. 
\end{itemize}

The \emph{extended} variant of MBQC refers to the possibility to perform measurements in the three possible planes $\{X,Y\}$, $\{X,Z\}$ and $\{Y,Z\}$ of the Bloch sphere, whereas all measurements are performed in the $\{X,Y\}$-plane in the original measurement-based quantum computation.  

\section{Reversibility, Determinism, and Generalized Flow}

Despite of the probabilistic evolution of quantum measurements, the correction mechanism can be used to make the overall evolution of an MBQC reversible which means that there exists an isometry $U$ ($U^\dagger U=\mathbb I$) from the input to the output qubits such that, whatever the classical outcomes of the measurements during the computation are, the evolution implemented by the MBQC is $U$.  In the context of measurement-based quantum computation this form of reversibility is called \emph{determinism} \cite{DK}. 
Determinism is an essential feature which is used for instance for proving that any quantum circuit can be simulated by an MBQC. Thus, this is a key ingredient for the universality of the model for quantum computing. The existence of a correction strategy that makes an MBQC deterministic crucially depends on the initial entangled state, i.e. on the open graph $(G,I,O)$ and the planes of measurement: given $\lambda: \comp{O}\to \{\{X,Y\},\{X,Z\},\{Y,Z\}\}$ a map which associates with every non output qubit its  plane of measurement, an \emph{extended open graph} $(G,I,O,\lambda)$ is \emph{uniformly deterministic} if for any measurement angles $\alpha:\comp{O}\to [0,2\pi)$, there exist two corrective maps $\mathtt x$ and $\mathtt z$ such that the corresponding MBQC is deterministic.  

Significant efforts have been made to characterize the open graphs that guarantees uniform determinism. Flow \cite{DK}, and generalised flow (\emph{gflow}) \cite{BKMP} are graphical conditions which are sufficient for uniform determinism.  
Gflow can be defined as follows for the extended open graphs:



\begin{definition}[GFlow]
An extended open graph $(G,I,O,\lambda)$ has a gflow if there exists $g:\comp{O}\to 2^{\comp{I}}$ s.t. $u\mapsto g(u)\cup Odd(g(u))$ is extensive and 
  for any $u \in \comp{O}$, 
\begin{eqnarray*}
\lambda(u) = \{X,Y\} &\Rightarrow &u\in Odd(g(u))\setminus g(u)\\
\lambda(u) = \{X,Z\} & \Rightarrow &u\in g(u) \cap Odd(g(u))\\
\lambda(u) = \{Y,Z\}&\Rightarrow &u\in g(u)\setminus Odd(g(u))
\end{eqnarray*}
where $Odd(A)=\{w\in V~|~|N(w)\cap A| = 0\bmod 2\}$ is the \emph{odd neighbourhood} of $A$ 
and a map   $f:\comp{O}\to 2^V$ is \emph{extensive} if there exists a partial order $\prec$ such that for any $u\in \comp{O}$, $u$ is smaller than its image by $f$ i.e., $\forall v\in V\setminus \{u\}, v\in f(u)\Rightarrow u\prec v$. 
\end{definition}


Concretely, if an extended open graph $(G,I,O)$ has a gflow $g$ then for any measurement angles $\alpha:\comp{O}\to [0, 2\pi)$ the corrective maps defined as $\forall u\in \comp{O}, \mathtt x(u):= g(u)\setminus \{u\}$ and $\mathtt z(u):= Odd(g(u)) \setminus \{u\}$ guarantees that the corresponding MBQC is deterministic \cite{BKMP}. 

With some additional assumptions gflow is not only sufficient but also necessary for determinism in measurement-based quantum computing. More precisely, there are mainly two cases to consider, depending on the number of inputs and outputs of the computation. When there are as many inputs as outputs, determinism corresponds to the notion of unitary evolution (evolution $U$ s.t. $U^\dagger U=UU^\dagger = \mathbb I$). In this particular case,  
 the gflow condition is necessary for strong -- i.e., all measurements occur with the same probability -- uniform determinism \cite{MMPST}. 
In the general case, when the number of inputs and outputs may differ, determinism corresponds to isometries (also called \emph{unitary embedding}). In this general case,  gflow characterizes \emph{stepwise} strong uniform determinism (roughly speaking the additional stepwise condition means that any partial computation is also deterministic) \cite{BKMP}. Notice that it is not known whether the strong and stepwise conditions are required: there is no known example of uniformly deterministic MBQC which corresponding open graph does not have a gflow.

%

Notice that if an extended open graph has a gflow then all the input qubits must be measured in the $\{X,Y\}$-plane:

\begin{property}\label{prop:XY}
If an extended open graph $(G,I,O,\lambda)$ has a gflow then $\forall u\in I\cap \comp{O}$, $\lambda(u) = \{X,Y\}$. 
\end{property}

\begin{proof} Let $g$ be a gflow for $(G,I,O,\lambda)$, and $u\in I\cap \comp{O}$, since for any $u\in \comp{O}$, $g(u)\subseteq \comp{I}$, $u\notin g(u)$, thus according to the definition of gflow, $\lambda(u)\neq \{X,Z\}$ and $\lambda(u)\neq \{Y,Z\}$. \hfill $\Box$ 
\end{proof}

\section{Focused Gflow and Normal Forms}

The gflow of an (extended) open graph is not unique in general. In the non extended case i.e., when all measurements are performed in the $\{X,Y\}$-plane several classes of gflow have been identified: 
%
%
%
%
%
%
%
the \emph{maximally delayed gflow} which depth is minimal and which is produced by an polytime algorithm \cite{MP08}; and the focus gflow which guarantees that the $\mathtt z$ corrective map acts only on the output qubits. The definition of focused gflow is as follows: Given an open graph $(G,I,O)$, a gflow $g$ is \emph{focused} if $\forall u\in \comp{O}$,  $Odd(g(u))\cap \comp{O} = \{u\}$. Since any gflow can be transformed into a focused gflow \cite{MMPST}, focused gflow can be used to characterize the open graphs that have a gflow:

\begin{property}
An open graph $(G,I,O)$ has a gflow if and only if there exists $g: \comp{O} \to 2^{\comp{I}}$ extensive such that
$\forall u\in \comp{O}$, $$Odd(g(u))\cap \comp{O} = \{u\}$$
\end{property}

Focused gflow is a simpler but  equivalent variant of gflow, which can be used for instance as a tool for quantum circuits translation and optimisation \cite{BK,DP,DPK}.

So far, there is no definition of `focused' gflow in the context of the extended MBQC. By symmetry, there are three natural kinds of `focused' extended gflow: those for which $Odd(g(u))\cap \comp{O} \subseteq \{u\}$; those for which $g(u)\cap \comp{O} \subseteq \{u\}$; and finally those for which $g(u)\oplus Odd(g(u))\cap \comp{O} \subseteq \{u\}$, $\oplus$ denotes the \emph{symmetric difference}. We define the corresponding three normal forms (NF for short) for extended gflows:

\begin{definition}[Normal forms] 
A gflow $g$  of an extended open graph $(G,I,O,\lambda)$ is
\begin{itemize}
\item $X$-NF  if $\forall u\in \comp{O}$, $$Odd(g(u)) \subseteq  \{u\}\cup O$$
\item $Y$-NF  if $\forall u\in \comp{O}$, $$\left(Odd(g(u))\oplus g(u) \right)  \subseteq  \{u\}\cup O$$
\item $Z$-NF  if $\forall u\in \comp{O}$, $$g(u)  \subseteq  \{u\}\cup O$$
\end{itemize}
\end{definition}

Intuitively a $\sigma$-NF, for $\sigma\in \{X,Y,Z\}$, guarantees that in the corresponding MBQC all the correctors applied on the non output qubits are Pauli-$\sigma$ operators. For instance, given a Z-NF gflow, in the corresponding MBQC $\forall u\in \comp{O}, \mathtt x(u)=g(u)\setminus \{u\}\subseteq O$ which implies that all Pauli correctors applied on non output qubits are $Z$ operators. Given a Y-NF gflow, in the corresponding MBQC $\forall u\in \comp{O}, \mathtt x(u)\cap \comp{O} = \mathtt z(u)\cap \comp{O}$ which means that all the Pauli correctors applied on non output qubits are products of $X$ and $Z$ which is nothing but Pauli-Y operators (up to a global phase).  
 Notice that given an open graph $(G,I,O)$, $g$ is a focused gflow of $(G,I,O)$ if and only if $g$ is a X-NF gflow of $(G,I,O,u\mapsto \{X,Y\})$. 


\section{Existence of Normal Forms}
In this section we consider the problem of the existence of gflow in normal forms. First notice that some extended open graphs  have a gflow but no $Z$-NF gflow for instance. The following extended open graph $(G,I,O,\lambda)$ where $G=(\{1,2,3\}, \{(1,2),(2,3)\})$, $I=\{1\}$, $O=\{3\}$ and $\lambda(1)= \lambda (2) = \{X,Y\}$ admits exactly two gflows $g$ and $g'$ ($g(1) = \{1\}$, $g'(1)=\{2,3\}$, and $g(2)=g'(2)=\{3\}$), none of them is in the Z-normal form.  

\begin{psfrags}
   \psfrag{1}{$1$}
      \psfrag{2}{$2$}
         \psfrag{3}{$3$}
               \psfrag{XY}[c]{~~~~~$\{X,Y\}$}
\begin{center}
\includegraphics[scale=0.8]{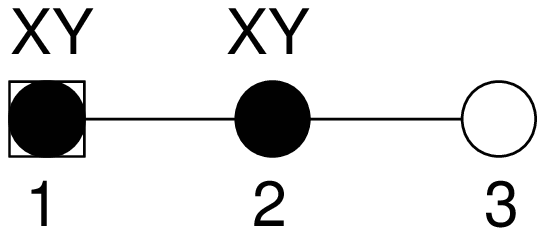}
\end{center}
\end{psfrags}

This simple example points out  a crucial difference with respect to  the non-extended case for which any gflow can be turned into a focused gflow. 
A sufficient condition for the existence of a $\sigma$-NF gflow for an extended  open graph with gflow is that every non-input measurement plane  contains $\sigma$:

\begin{theorem}\label{thm:sc}
If an extended open graph $(G,I,O,\lambda)$ has a gflow then, for any $\sigma\in \bigcap_{u\in \comp{I}\cap \comp{O}} \lambda(u)$, $(G,I,O,\lambda)$ has a  $\sigma$-NF gflow. 
 \end{theorem}
 
 \begin{proof}
 Let $g$ be a gflow for $(G,I,O,\lambda)$, and $\sigma\in  \bigcap_{u\in \comp{I}\cap \comp{O}} \lambda(u)$. We define $g_\sigma:\comp{O}\to 2^{\comp{I}}$ as follows, depending on $\sigma$:
\begin{eqnarray*}
g_X(u)&:=& g(u)\oplus \left(\bigoplus_{{v\in Odd(g(u))\setminus (O\cup \{u\})}}g_X(v)\right)\\
  g_Y(u)&:=& g(u)\oplus \left(\bigoplus_{{v\in (g(u)\oplus Odd(g(u)))\setminus (O\cup \{u\})}}g_Y(v)\right)\\
  g_Z(u)&:=& g(u)\oplus \left(\bigoplus_{{v\in g(u)\setminus (O\cup \{u\})}}g_Z(v)\right)\\
  \end{eqnarray*}
  Extensivity of $g$ guarantees that $g_\sigma$ is well-defined. In the following we prove that $g_\sigma$ is a gflow, and then that $g_\sigma$ is in $\sigma$-NF. \\{[gflow]} Let $\prec$ a partial order according to which $g$ is extensive, we show that $g_\sigma$ is also extensive according to $\prec$. Indeed, for any $u\in \comp{O}$ and any $w\in V\setminus \{u\}$, s.t. $w\in g_\sigma(u)$, by induction if there is no larger elements in $\comp{O}$ then $g_\sigma(u) = g(u)$, so $u\prec w$. Otherwise, $w\in g(u)\cup (\bigcup_{v\in g(u)\cup Odd(g(u))\setminus (O\cup \{u\})}g_\sigma(v))$, so either (i) $w\in g(u)$ which implies $u\prec w$, or (ii) $\exists v\in g(u)\cup Odd(g(u))$ s.t. $w\in g_\sigma(v)$, so $u\prec v$ and, by induction, $v\prec w$ which implies $u\prec w$. \\
Regarding the remaining gflow conditions, notice that the extensivity of $g$ and $g_\sigma$ guarantees that for any $u\in \comp{O}$, $g_\sigma(u)\cap \{u\} = g(u)\cap \{u\}$ and $Odd(g_\sigma(u))\cap \{u\} = Odd(g(u))\cap \{u\}$ (the linearity of $Odd$ is also used in this second case: $Odd(A\oplus B) = Odd(A)\oplus Odd(B)$). Thus $g_\sigma$ is a gflow.
\\{[$\sigma$-NF]} In the following we prove that $g_\sigma$ is in a $\sigma$-NF. W.l.o.g. assume $\sigma=Y$ (the  other two cases are similar). We actually prove by induction that $\forall u \in \comp{O}$, $Odd(g_Y(u)\oplus g_Y(u)) \cap \comp{O}= \{u\}$. Let $u\in \comp{O}$.
\begin{itemize}
\item  If there is no larger element according to $\prec$ (the partial order induced by $g$ and $g_Y$) in $\comp{O}$, then $Odd(g_Y(u))\oplus g_Y(u)\subseteq Odd(g_Y(u))\cup g_Y(u)\subseteq \{u\}\cup O$ by extensivity of $g_Y$, moreover since $Y\in \lambda(u)$, $u\in Odd(g_Y(u))\oplus g_Y(u)$, so $(Odd(g_Y(u))\oplus g_Y(u)) \cap \comp{O} = \{u\}$. 
\item Otherwise, $(Odd(g_Y(u))\oplus g_Y(u))\cap \comp{O}=$
\begin{eqnarray*}
&& \left(Odd(g(u))\oplus g(u) \oplus \left(\hspace{-2.8cm}\bigoplus_{{\hspace{2.8cm}v\in (g(u)\oplus Odd(g(u)))\setminus (O\cup \{u\})}} \hspace{-2.8cm}Odd(g_Y(v))\oplus g_Y(v)\right)\right)\cap \comp{O}\\
&=&\left(Odd(g(u))\oplus g(u)\right)\cap \comp{O} \oplus \left(\hspace{-2.8cm}\bigoplus_{{\hspace{2.8cm}v\in (g(u)\oplus Odd(g(u)))\setminus (O\cup \{u\})}} \hspace{-2.8cm}(Odd(g_Y(v))\oplus g_Y(v))\cap \comp{O}\right)\\
&=&\left(Odd(g(u))\oplus g(u)\right)\cap \comp{O} \oplus \left(\hspace{-2cm}\bigoplus_{{\hspace{2cm}v\in (g(u)\oplus Odd(g(u)))\setminus (O\cup \{u\})}} \hspace{-2cm}\{v\}~~~~~~~~~~~~~\right)\\
&=&\left(Odd(g(u))\oplus g(u)\right)\cap \comp{O} \oplus \left((g(u)\oplus Odd(g(u)))\setminus (O\cup \{u\})\right)\\
&=&\left(Odd(g(u))\oplus g(u)\right)\cap \{u\}
\end{eqnarray*}
Moreover, since $Y\in \lambda(u)$, $u\in Odd(g(u))\oplus g(u)$, so $(Odd(g_Y(u))\oplus g_Y(u))\cap \comp{O}=\{u\}$. \hfill $\Box$
\end{itemize}

  \end{proof}
 
 As a corollary, any (non extended) open graphs with gflow, admits both X- and Y-NF gflows. More generally, any extended open graph $(G,I,O,\lambda)$ with gflow such that $\lambda$ is constant over $\comp{I}\cap \comp{O}$ admits  both $\sigma$- and $\sigma'$-NF gflows where $\comp{I}\cap \comp{O}\subseteq \lambda^{-1}(\{\sigma,\sigma'\})$
 
 Theorem \ref{thm:sc} provides a sufficient condition for the existence of a $\sigma$-normal form. The following example points out that this condition is not necessary: in this extended open graph $\lambda(2)=\{X,Z\}$ however it admis the following Y-NF gflow $1\mapsto \{4\} ; 2\mapsto \{2,3,4\}$.   
 
 \begin{psfrags}
   \psfrag{1}{$1$}
      \psfrag{2}{$2$}
         \psfrag{3}{$3$}
                  \psfrag{4}{$4$}
               \psfrag{XY}[c]{~~~~~$\{X,Y\}$}
                         \psfrag{YZ}[c]{~~~~~$\{X,Z\}$}
\begin{center}
\includegraphics[scale=0.8]{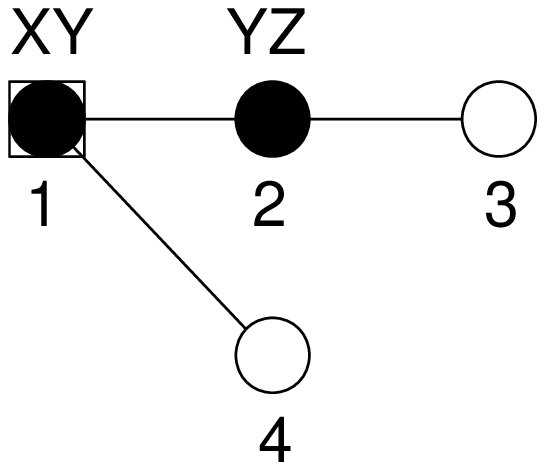}
\end{center}
\end{psfrags}

 Notice that in this counter example there are strictly more outputs than inputs.  %
 Indeed, we show that the existence of a $\sigma$-NF gflow with $\sigma \in \{Y,Z\}$, implies that the number non-input measurement-planes which do not contain $\sigma$ is upper bounded by the \emph{input defect} i.e., the difference between the number of outputs and inputs: 

 
%
%
%
%

\begin{theorem}\label{thm:cn}
 Given $\sigma\in \{Y,Z\}$ and an extended open graph $(G,I,O,\lambda)$, if $(G,I,O,\lambda)$ has a $\sigma$-NF gflow then $$|\{u \in \comp I \cap \comp O ~|~ \sigma\notin \lambda(u)\}| \le |O|-|I|$$
 \end{theorem}

\begin{proof}
Given $(G,I,O,\lambda)$ with a $\sigma$-NF gflow $g$ where $\sigma\in \{Y,Z\}$, we show that any non-input vertex which is measured in a plane which does not contain  $\sigma$ can be, roughly speaking, turned into an input vertex. The proof is by induction on $|\{u \in \comp I \cap \comp O ~|~ \sigma\notin \lambda(u)\}|$. If $|\{u \in \comp I \cap \comp O ~|~ \sigma\notin \lambda(u)\}| = 0$ the property is satisfied since determinism implies   $|I|\le |O|$. 
Otherwise, let $u_0\in \comp I\cap \comp O$ s.t. $\sigma\notin \lambda(u_0)$ and let $g'(u):=\begin{cases} g(u)&\text{if $u=u_0$ or $u_0\notin g(u)$}\\ g(u)\oplus g(u_0)&\text{otherwise}\end{cases}$. $g'$ is a $\sigma$-NF gflow s.t. $\forall u\in \comp O \setminus \{u_0\}$, $u_0\notin g'(u)$. 
\\{[Z-NF]} If $\sigma=Z$, $\lambda(u_0) = \{X,Y\}$, so $u_0\notin g'(u_0)$. As a consequence $\forall  u\in \comp O, g'(u)\in \comp{(I\cup \{u_0\})}$, and $g'$ is a Z-NF gflow of $(G,I\cup \{u_0\},O, \lambda)$: in this new extended open graph the number of measurement-planes which do not contain $Z$ is decreased by one, as well as the input defect i.e., the difference between the number of outputs and inputs. 
\\{[Y-NF]} If $\sigma=Y$, a new degree-one vertex $u_1$ is connected to $u_0$, and let $g'':\comp O\to 2^{\comp {(I\cup \{u_0\})}}$ be defined as follows\\\centerline{$g''(u):=\begin{cases}\{u_1\}&\text{if $u=u_0$}\\ g'(u_0)\oplus \{u_0, u_1\} &\text{if $u=u_1$}\\
g'(u)&\text{otherwise}\end{cases}$}
$g''$ is a Y-NF gflow for $(G',I\cup\{u_0\},O, \lambda')$, where $G'$ is the graph $G$ augmented with the dangling vertex $u_1$, and $\lambda'(u)=\begin{cases}\{X,Y\}&\text{if $u=u_0$}\\\{Y,Z\}&\text{if $u=u_1$}\\\lambda(u)&\text{otherwise}\end{cases}$. In this new open graph the number of inputs is increased by one, so the input defect decreases by one, moreover the number of measurement planes which do not contain $Y$ also decreases by one since $u_1$ is measured in the $\{Y,Z\}$-plane in this new open graph. \hfill $\Box$
\end{proof}

 \begin{corollary} Given $\sigma\in \{Y,Z\}$ and an extended open graph $(G,I,O,\lambda)$ with gflow such that $|I| = |O|$, $(G,I,O,\lambda)$ has a $\sigma$-NF gflow if and only if for any $ u\in \comp{I}\cap \comp{O}$, $\sigma\in\lambda(u)$. 
 \end{corollary}

 Theorem \ref{thm:cn} shows that in a Z-NF gflow, when a non-input  is measured in the $\{X,Y\}$-plane, this non-input  somehow behaves as an input. Regarding the Y-NF gflow when a non-input qubit is measured in the $\{X,Z\}$-plane, this qubit cannot be seen as an input qubit mainly because all inputs have to be measured in the $\{X,Y\}$-plane (Property \ref{prop:XY}). However, up to a  transformation of the graph, it can be turned into an input (see proof of Theorem  \ref{thm:cn}). One can wonder whether such a transformation exists for X-NF gflow? Surprisingly,  Theorem \ref{thm:cn} cannot be extended to the X-NF case as illustrated by the following counter example where the number of inputs is equal to the number of outputs and which has a X-NF gflow ($1\mapsto \{3\} ; 2\mapsto \{2,3\}$) despite of the measurement of a non-input qubit in the $\{Y,Z\}$-plane:

  \begin{psfrags}
   \psfrag{1}{$1$}
      \psfrag{2}{$2$}
                  \psfrag{4}{$3$}
               \psfrag{XY}[c]{~~~~~$\{X,Y\}$}
                         \psfrag{YZ}[c]{~~~~~$\{Y,Z\}$}
\begin{center}
\includegraphics[scale=0.8]{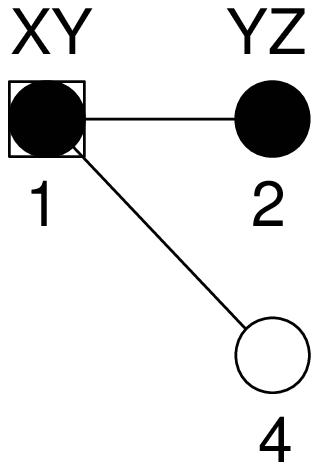}
\end{center}
\end{psfrags}

%
%
%
\bibliographystyle{plain}

\end{document}